\documentclass[reqno,a4paper,10pt]{amsart}

\usepackage{graphicx}         
\usepackage{amsmath}
\usepackage{amsfonts}
\usepackage{amssymb}
\usepackage{eucal}
\usepackage[latin1]{inputenc}
\usepackage[all]{xy}

\newtheorem{thm}{Theorem}[section]
\newtheorem{cor}[thm]{Corollary}
\newtheorem{lemma}[thm]{Lemma}

\newtheorem{prop}[thm]{Proposition}
\theoremstyle{definition}
\newtheorem{defn}[thm]{Definition}
\theoremstyle{remark}

\newtheorem{example}[thm]{Example}

\renewcommand{\l}{\lambda}
\newcommand{\s}{\sigma}

\newcommand{\PB}{\left\{\cdot\,,\cdot\right\}}

\newcommand{\Pb}[1]{\left\{\cdot\,,#1\right\}}
\newcommand{\pb}[1]{\left\{#1\right\}}

\renewcommand{\)}{\right)}
\renewcommand{\[}{\left[}
\renewcommand{\]}{\right]}

\newcommand{\set}[1]{\left\{#1\right\}}

\newcommand{\cK}{\mathcal K}

\newcommand{\X}{\mathcal X}

\newcommand{\bbR}{\mathbb R}

\newcommand{\F}{\mathbf F}

\newcommand{\hcS}[2]{\hat{\mathcal S}^{(#1)}_{#2}}
\newcommand{\cS}[2]{\mathcal S^{(#1)}_{#2}}
\newcommand{\um}{{\underline m}}
\newcommand{\hum}{\hat{\underline m}}

\newcommand{\xx}[1]{x_{#1}}
\newcommand{\pin}[2]{\pi_{#2}}
\newcommand{\eps}[2]{\epsilon^{#1}_{#2}}
\newcommand{\LV}{\hbox{LV}}

\newcommand{\ds}{\displaystyle}

\newcommand{\leqs}{\leqslant}
\newcommand{\geqs}{\geqslant}
\newcommand{\pp}[2]{\frac{\partial#1}{\partial#2}}
\newcommand{\p}{\partial}

\newcommand{\diff}{{\rm d }}

\newcommand{\Id}{\mathop{\rm Id}}

\renewcommand{\geq}{\geqs}
\renewcommand{\leq}{\leqs}

\newif\ifprivate
\privatefalse

 \numberwithin{equation}{section}

\def\???{\ifprivate {\bf {???}} \marginpar{{\Huge {\bf ?}}}\else \fi}
\numberwithin{equation}{section}

\begin{document}

\nocite{*}

\parskip 4pt
\baselineskip 16pt


\title[Integrable reductions of some Lotka-Volterra systems]{Integrable reductions of the Bogoyavlenskij-Itoh Lotka-Volterra systems}

\author{P. A. Damianou, C. A. Evripidou, P. Kassotakis}
\address{Department of Mathematics and Statistics\\
University of Cyprus\\
P.O.~Box 20537, 1678 Nicosia\\Cyprus}
\email{damianou@ucy.ac.cy, cevrip02@ucy.ac.cy, pavlos1978@gmail.com}
\author{P. Vanhaecke}
\address{ Pol Vanhaecke,  Laboratoire de Math\'ematiques\\
          UMR 7348 du CNRS\\
          Universit\'e de Poitiers\\
          86962 Futuroscope Chasseneuil Cedex\\
          France}
\email{ pol.vanhaecke@math.univ-poitiers.fr}

\thanks{Corresponding author: Pantelis  A. Damianou,  Email: damianou@ucy.ac.cy}
\date{\today}

\begin{abstract}
	Given a constant skew-symmetric matrix A, it is a difficult open problem
	whether the associated Lotka-Volterra system is integrable or not. We solve
	this problem in the special case  when A is a Toepliz matrix where all
	off-diagonal entries are plus or minus one. In this case, the associated
	Lotka-Volterra system turns out to be a reduction of   Liouville integrable
	systems, whose integrability was shown by Bogoyavlenskij and Itoh. We prove
	that the reduced systems are also Liouville integrable and that they are
	also non-commutative integrable by constructing a set of independent first
	integrals, having the required involutive properties (with respect to the
	Poisson bracket).  These first integrals fall into two categories.  One set
	consists of polynomial functions which can be obtained by a matricial
	reformulation of Itoh's combinatorial description. The other set consists
	of rational functions which are obtained through a Poisson map from the
	first integrals of some recently discovered superintegrable Lotka-Volterra
	systems. The fact that these polynomial and rational first integrals,
	combined, have the required properties for Liouville and non-commutative
	integrability is quite remarkable; the quite technical proof of functional
	independence of the first integrals is given in detail.
\end{abstract}
\subjclass[2010]{37J35, 39A22}

\keywords{Integrable systems, reduction, discretization}

\maketitle

\tableofcontents

\section{Introduction} \label{intro}
The Lotka-Volterra model is a basic model of predator-prey interactions. The model was developed independently by
A. Lotka \cite{Lotka}, and V. Volterra \cite{Volterra}.  It forms the basis for many models used today in the
analysis of population dynamics.

The most general form of Lotka-Volterra equations in dimension $n$ is
\begin{equation}\label{eq:LV_gen_intro}
  \dot x_i = \varepsilon_i x_i + \sum_{j=1}^n A_{i,j} x_i x_j, \ \ i=1,2, \dots , n \; .
\end{equation}
By now, many systems of the form (\ref{eq:LV_gen_intro}) have been introduced and studied,
often from the point of (Liouville, Darboux  or algebraic) integrability
\cite{Bog1,Bog2,bermejo,veselov,suris,fernan_pol,koul_quispel_pol_2016,const_damianou,bountis_pol} or Lie theory
\cite{Bog1,Bog2,damianou,ballesteros,char_damian_evrip}, but also in relation
with other integrable systems \cite{KKQTV,damian_fernan}.

For the systems which will be considered here, all constants $\varepsilon_i$ are zero (no linear terms) and the
constant matrix $A$ is skew-symmetric. It is well-known that (\ref{eq:LV_gen_intro}) is then a Hamiltonian system
with Poisson structure defined by
\begin{equation}\label{eq:poisson_intro}
  \pb{x_i,x_j}:=A_{i,j}x_ix_j\;,
\end{equation}
and Hamiltonian function $H:=x_1+x_2+\dots+x_n$. We will, more precisely, only be concerned in this paper with the
$n$ skew-symmetric matrices $A_0,\dots,A_{n-1}$ of the Toeplitz\footnote{Recall that a Toeplitz matrix is a matrix
  in which each descending diagonal from left to right is constant; when such a matrix is skew-symmetric, it is
  entirely determined by its first row.} form
\begin{equation}\label{eq:bogo_mat}
  A_k=
  \begin{pmatrix}
    0&1&1&\cdots&1&-1&-1&\cdots&-1&-1\\
   -1&0&1&\cdots&1&1&-1&\cdots&-1&-1\\
   -1&-1&0&\cdots&1&1&1&\ddots&-1&-1\\
   \vdots&\vdots& & \ddots & \vdots & \vdots & \vdots & \ddots & \ddots & \vdots \\
   -1&-1&-1&\cdots&\cdots&\cdots&\cdots&\cdots&1&-1\\
    1&-1&-1&\cdots&\cdots&\cdots&\cdots&\cdots&1&1\\
    \vdots&\vdots& & \ddots & \vdots & \vdots & \vdots & \ddots & \vdots & \vdots  \\
    1&1&1&\cdots&-1&-1&-1&\cdots&0&1\\
    1&1&1&\cdots&1&-1&-1&\cdots&1&0\\
  \end{pmatrix},
\end{equation}
with $-1$ appearing $k$ times on the first row. The size of the matrix $A_k$ is $n$, which we sometimes indicate
explicitly by writing $A^{(n)}_k$ for $A_k$. Also, the Poisson structure which corresponds to $A_k$, as in
(\ref{eq:poisson_intro}), is denoted by $\pi_k$ or $\pi^{(n)}_k$. The corresponding Lotka-Volterra system
(\ref{eq:LV_gen_intro}) will be denoted by $\LV(n,k)$.

Two families of Lotka-Volterra systems $\LV(n,k)$ have already been studied from the point of view of
integrability. The first one, which we will refer to as the Bogoyavlenskij-Itoh case, is when $n=2k+1$. Notice that
$A_k$ is then a circulant\footnote{A circulant matrix of size $n$ is a Toeplitz matrix $A$ satisfying the
  additional property that $A_{i,n}=A_{i+1,1}$ for $i=1,\dots n-1$, so that each row is obtained from the previous
  row by rotating it by one element to the right.} matrix and the system has a symmetry of order $n$, given by
permuting the variables in a cyclic way. In \cite{itoh1}, Y. Itoh gives explicit combinatorial formulas for $k+1$
independent first integrals $K_0,K_1,\dots,K_k$ of $\LV(2k+1,k)$, where $K_i$ is a homogeneous polynomial of degree
$2i+1$; in particular, $K_0$ is the linear Hamiltonian $H$. An alternative construction of these first integrals
was given in \cite{Bog1} by O. Bogoyavlenskij, who obtains them as spectral invariants of a Lax operator which he
constructs. Next, Y. Itoh shows in \cite{itoh2} by a beautiful combinatorial argument that the integrals
$K_0,K_1,\dots,K_k$ are pairwise in involution (Poisson commute). Since the rank of the Poisson structure
$\pi_k^{(2k+1)}$ is $2k$, this shows that $\LV(2k+1,k)$ is integrable in the sense of Liouville, for all $k$.

More recently, another family of Lotka-Volterra systems came up in the study of some polynomials (so-called
multi-sums of products) which appear as invariants of a discretization of some integrable equations, such as the
modified Korteweg-de Vries equation. This family consists of all $\LV(n,0)$, i.e., they correspond to the matrix
$A_0$, whose upper-triangular entries are all equal to $1$. It was shown in \cite{KKQTV} that these systems have
$\left[\frac{n+1}2\right]$ independent first integrals which are pairwise in involution. Again, this shows that
$\LV(n,0)$ is integrable in the sense of Liouville, since the rank of the Poisson structure $\pi_0^{(n)}$ is $n$
when $n$ is even, and $n-1$ otherwise. In addition, it is shown in~\cite{KKQTV} that $\LV(n,0)$ is also
superintegrable, i.e., it has $n-1$ independent (rational) first integrals. This alternative viewpoint of the
integrability of these systems exhibits the integral curves of the Hamiltonian vector field (\ref{eq:LV_gen_intro})
as being confined to tori which are of lower dimension than what is expected from Liouville integrability. This
property has important implications to the dynamics of the Hamiltonian system.

The starting point of the present paper is the observation that $\LV(n,0)$ is a reduction of the
Bogoyavlenskij-Itoh system $\LV(2n-1,n-1)$: setting the last $n-1$ variables of the latter system equal to zero, we
get a Poisson submanifold linearly isomorphic to $\bbR^n$, the restricted Poisson structure is $\pi_0$ and the
Hamiltonian of $\LV(2n-1,n-1)$, restricted to the submanifold, is precisely the Hamiltonian of $\LV(n,0)$. This
does not mean that the Liouville integrability of $\LV(n,0)$ is a consequence of the Liouville integrability of
$\LV(2n-1,n-1)$; on the contrary, except for the Hamiltonian $H=K_0$ each one of the first integrals $K_i$ becomes
trivial (zero) under the reduction; in particular, the rational integrals of $\LV(n,0)$ cannot be obtained from the
polynomial first integrals of $\LV(2n-1,n-1)$. The natural question which arizes is the integrability of the
systems that interpolate between $\LV(2n-1,n-1)$ and $\LV(n,0)$. In fact, it is easy to see that starting from
$\LV(2n-1,n-1)$ and setting successively the last surviving variable equal to zero, one gets the following string
of Lotka-Volterra systems:
\begin{equation*}
  \LV(2n-1,n-1)\to \LV(2n-2,n-2)\to\cdots\to\LV(n+1,1)\to \LV(n,0)\;,
\end{equation*}%
with corresponding Poisson structures $\pi_{n-1},\ \pi_{n-2},\dots,\pi_1,\pi_0$ (in the appropriate dimensions). In
each step, precisely one of the polynomial first integrals becomes trivial (namely, the one of highest degree), yet
we will show that these Lotka-Volterra systems are Liouville integrable by constructing, at each step, a sufficient
number of independent rational first integrals, which are themselves pairwise in involution, but are also in
involution with the (restricted) polynomial first integrals. But what happens with superintegrability?
Non-commutative integrability, which interpolates between Liouville integrability and superintegrability is the
answer! Quickly stated (see Definition~\ref{def:non-com}
below for a precise formulation and also \cite{camil_miranta_pol,PLV}), a Hamiltonian system on
an $n$-dimensional Poisson manifold is a non-commutative integrable system of rank $r$ if it has $n-r$ independent
first integrals, $r$ of which are in involution with all $n-r$ first integrals (so the Hamiltonian is among
them). Clearly, superintegrability corresponds to $r=1$; also, Liouville integrability correspond to the case in
which $r$ is half the rank of the Poisson manifold (\emph{all} $n-r$ first integrals are then pairwise in
involution).

We can now state the main theorem of this paper. Fix $n$ and $k$ with $n>2k+1$. For $i=0,1,\dots,k$ let
$K_i^{(n,k)}$ denote the restriction of the polynomial first integral $K_i$ of $\LV(2n-2k-1,n-k-1)$ to $\LV(n,k)$.
Also, for $\ell=1,\dots,n-2k-2$ denote by $H^{(n,k)}_\ell$ the $n-2k-2$ rational\footnote{The pullback of the
  Hamiltonian $H$ of $\LV(n-2k,0)$ is excluded from this list because it is equal to $K_k^{(n,k)}$.} first
integrals of $\LV(n-2k,0)$, pulled back to $\LV(n,k)$ (using the Poisson map in Proposition
\ref{prop:poisson_map}).
\begin{thm}\label{thm:main}
  Consider the Lotka-Volterra system $\LV(n,k)$, where $n>2k+1$.
  \begin{enumerate}
  \item[(1)] It is non-commutative integrable of rank $k+1$, with first integrals
    \begin{equation}\label{eq:first_int_intro}
      H=K_0^{(n,k)},K_1^{(n,k)}\dots,K_k^{(n,k)},H^{(n,k)}_1,H^{(n,k)}_2,\dots, H^{(n,k)}_{n-2k-2}\;.
    \end{equation}
    The first $k+1$ functions of this list have independent Hamiltonian vector fields and are in involution with
    every function of the complete list (\ref{eq:first_int_intro}).
  \item[(2)] It is Liouville  integrable with first integrals
    $$H=K_0^{(n,k)},K_1^{(n,k)}\dots,K_k^{(n,k)},H^{(n,k)}_1,H^{(n,k)}_2,
    \dots,H^{(n,k)}_{r-1}\;,$$ where $r:=\left[\frac{n+1}2\right]-k.$
  \end{enumerate}

\end{thm}
%

\section{Diagonal Poisson structures and Poisson maps}
We first introduce the Poisson structures which appear in the Lotka-Volterra systems which we will construct in the
next section. For any $k$ with $0\leqslant k<n$ we define a skew-symmetric Toeplitz matrix $A_k^{(n)}$ of size $n$
by setting, for $1\leqs i<j\leqs n$,
\begin{equation}\label{eq:A_k_def_2}
  \(A_k^{(n)}\)_{i,j}:=\eps{n+i}{k+j}\;\quad\text{where}\quad\eps m\ell:=
  \begin{cases}
     \;\;1 \qquad    m>\ell\;,\\
    -1 \qquad    m\leqs\ell\;.
  \end{cases}
\end{equation}%
It is fully determined by its first row, which is given by $(0,1,1,\dots,1, -1,-1,\dots,-1)$, with $-1$ appearing
$k$ times (at the end). When its size is clear from the context, we also write $A_k$ for $A_k^{(n)}$, and similarly
for the entries of this matrix. Using $A_k$ we consider the quadratic Poisson structure $\pi^{(n)}_k=\pi_k$ on
$\bbR^n$, defined by the following brackets:
\begin{equation}\label{eq:pik_def_2}
  \pb{\xx i,\xx j}_k=\(A_k\)_{i,j}x_ix_j=\eps{n+i}{k+j}x_ix_j\;.
\end{equation}%
These quadratic Poisson structures are called \emph{diagonal}, because the Poisson bracket of $x_i$ and $x_j$ is a
scalar multiple of their product $x_ix_j$; it is well-known that diagonal brackets always satisfy the Jacobi
identity, hence they are indeed Poisson brackets. The rank of the Poisson
structures $\pi_k$ is given by the following elementary proposition.
\begin{prop}\label{prop:pi_rank}
  The rank of $\pi_k=\pi^{(n)}_k$ is $n$ when $n$ is even and $n-1$ when $n$ is odd. In the latter case,
  \begin{equation}\label{eq:Casimir}
    C:= x_1x_2\dots x_k\ds\frac{x_{k+1}x_{k+3}\ldots x_{n-k}}{x_{k+2}x_{k+4}\ldots x_{n-k-1}}x_{n-k+1}\dots
    x_{n-1}x_n\;
  \end{equation}%
  is a Casimir function of $\pi_k$.
\end{prop}
\begin{proof}
It is well-known (see e.g.\ \cite[Example~8.14]{PLV}) that the rank of the diagonal Poisson structure~$\pin nk$ (at
a generic point) is equal to the rank of its defining matrix~$A_k$. Let us first show that the rank of $A_k$ is $n$
when $n$ is even. To do this, we show that the determinant of $A_k$ is $1$ modulo $2$. This is done by replacing in
$A_k$ the $i$-th row by the sum (modulo~$2$) of its $i$-th and $(i+1)$-th rows, for $i=1,\dots,n-1$; also, we
replace the last row by the sum (modulo~$2$) of all the other rows of $A_k$. The resulting matrix is upper
triangular, with all its diagonal entries equal to $1$ modulo 2. This proves that when $n$ is even, $A_k$ is of
rank $n$. When $n$ is odd, $A_k$ cannot be of rank $n$ because $A_k$ is skew-symmetric, but the top left principal
minor of $A_k$ is invertible, since it is of the above form (modulo~$2$), hence the rank of $A_k$ is $n-1$. To
prove that $C$ is a Casimir of $\pi_k$ when $n$ is odd it suffices to show that $\pb{x_i,C}=0$ for $i=1,\dots,n$,
which is easily done by direct computation, using (\ref{eq:pik_def_2}). Alternatively, one checks using
(\ref{eq:A_k_def_2}) that the following vector
\begin{equation*}
  (\underbrace{1,1,\dots,1}_k,\underbrace{1,-1,1,-1,\dots,1,-1}_{n-2k-1},1,\underbrace{1,1,\dots,1}_k)
\end{equation*}%
is a null vector of $A_k$.
\end{proof}
We show in the following two propositions how the Poisson structures $\pi_k$ are related.
\begin{prop}\label{prp:inclusion}
  For $\ell=0,\dots,n$, consider the inclusion map
  \begin{equation}
    \begin{array}{lcccl}
      \imath_\ell&:&\bbR^n&\to&\bbR^{n+1}\\
                 & &(x_1,x_2,\dots,x_n) &\mapsto&(x_1,x_2,\dots,x_{\ell},0,x_{\ell+1},x_{\ell+2},\dots,x_n)\;.
    \end{array}
  \end{equation}
  For any $k$ with $0\leqs k\leqs n$, the linear subspace $\imath_\ell\bbR^n$ is a Poisson submanifold of
  $(\bbR^{n+1},\pin{n+1}k)$, and so $\imath_\ell$ is a Poisson map, when $\bbR^n$ is equipped with the reduced
  Poisson structure:
  \begin{enumerate}
    \item[a)] If $k<\ell\leqs n-k$, then the reduced Poisson structure (on $\imath_\ell\bbR^n\simeq\bbR^n$) is~$\pin
      {n}{k}$;
    \item[b)] If $k=0$, then the reduced Poisson structure is $\pin {n}{0}$;
    \item[c)] If $\ell=n$ and $k>0$, then the reduced Poisson structure is $\pin {n}{k-1}$.
  \end{enumerate}
  In each one of these cases, the Hamiltonian system on $(\bbR^{n+1},\pin{n+1}k)$ defined by a function $H$
  restricts on $\imath_\ell\bbR^n$ to a Hamiltonian system, with the restriction of $H$ as Hamiltonian.
\end{prop}
\begin{proof}
Recall that a submanifold $N$ of a Poisson manifold $(M,\Pi)$ is a \emph{Poisson submanifold} if all Hamiltonian
vector fields of $(M,\Pi)$ are tangent to $N$ (at points of $N$). In our case $N$ is the submanifold of
$M:=\bbR^{n+1}$, defined by $y_{\ell+1}=0$, where we denote by $y_1,\dots,y_{n+1}$ the standard coordinates on
$\bbR^{n+1}$. Let $F$ be a function on $M$ and consider its Hamiltonian vector field, which is given by
$\X_F:=\Pb{F}$. Thanks to the diagonal nature of the brackets $\pi_k^{(n+1)}$, we see that
$\X_F[y_{\ell+1}]=\pb{y_{\ell+1},F}=y_{\ell+1}G$, for some function $G$ on $\bbR^{n+1}$, hence the bracket vanishes
on $N$.  So, $\X_F$ is tangent to $N$ and $N$ is a Poisson submanifold of $(M,\Pi)$. If we denote by $p_\ell$ the
natural projection of $\bbR^{n+1}$ on $\imath_\ell\bbR^n\simeq\bbR^n$, then for any function $F$ on $N$ an
extension of $F$ to $M$ is given by $\tilde F:=F\circ p_\ell$, and so the reduced Poisson structure on $N$ is given
by the following brackets:
\begin{equation}\label{eq:poisson_red}
  \pb{x_i,x_j}_N:=\pb{x_i\circ p_\ell,x_j\circ p_\ell}_k^{(n+1)}\circ\imath_\ell\;,\qquad 1\leqs i,j\leqs n\;.
\end{equation}%
Since $\xx i\circ p_\ell=y_i$ when $i\leqs\ell$ and $\xx i\circ p_\ell=y_{i+1}$ when $\ell<i$, the right hand
side of (\ref{eq:poisson_red}) is given (for $i<j$) by
$$
\begin{array}{rc}
  \eps{n+i+1}{j+k+1}\xx i\xx j\;,\qquad &i\leqs\ell<j\;,\\
  \eps{n+i+1}{j+k}\xx i\xx j\;,\qquad &j\leqs\ell\hbox{ or } \ell<i\;.
\end{array}
$$
In cases a) and b) both formulas amount to $\eps{n+i}{j+k}\xx i\xx j=\pb{\xx i,\xx j}_k$, while they amount
 in case c) to $\eps{n+i}{j+k-1}\xx i\xx j=\pb{\xx i,\xx j}_{k-1}$.
\end{proof}
Notice that, since the reduced Poisson structure belongs again to our class of Poisson structures, the use of the
proposition can be repeated one or several times.  For example, as indicated in the introduction, one can by
repeated use of c) realize $\LV(n,0)$ as a Poisson reduction of $\LV(2n-1,n-1)$.
\begin{prop}\label{prop:poisson_map}
  For any $k$ with $0<2k<n$, the map defined by
\begin{equation}\label{eq:phi_def}
  \begin{array}{lcccl}
    \phi_k&:&(\bbR^n,\pin nk)&\to&(\bbR^{n-2k},\pin {n-2k}0)\\
                & &(x_1,x_2,\dots,x_n) &\mapsto&x_1x_2\dots x_k(x_{k+1},x_{k+2},\dots,x_{n-k})x_{n-k+1}\dots x_n\;,
  \end{array}
\end{equation}
is a Poisson map.
\end{prop}
\begin{proof}
Let us denote the natural coordinates on $\bbR^{n-2k}$ by $y_1,\dots,y_{n-2k}$. We need to show that
\begin{equation}\label{eq:phi_to_show}
  \pb{y_i,y_j}_0^{(n-2k)}\circ\phi_k=\pb{y_i\circ\phi_k,y_j\circ\phi_k}_k^{(n)}
\end{equation}
for all $i,j$ with $1\leqs i<j\leqs n-2k$. Let us denote by $P_k$ the product
of the first and last $k$ coordinates of $\bbR^n$, $P_k=x_1x_2\dots x_kx_{n-k+1}x_{n-k+2}\dots x_n$.
Then $y_i\circ\phi_k=P_kx_{i+k}$, and so the right hand side of (\ref{eq:phi_to_show}) is given by
$$
  \pb{P_kx_{i+k},P_kx_{j+k}}_k^{(n)}=  P_k^2\pb{x_{i+k},x_{j+k}}_k^{(n)}+
    x_{j+k}P_k\pb{x_{i+k},P_k}_k^{(n)}-x_{i+k}P_k\pb{x_{j+k},P_k}_k^{(n)}\;.
$$
The first term in this expression is the left hand side of (\ref{eq:phi_to_show}), since both are equal to
$P_k^2{x_{i+k}x_{j+k}}$ (no signs!); the second and third terms are both equal to zero, because
$\pb{x_{\ell},x_1x_2\dots x_k}_k^{(n)}=-kx_{\ell}x_1x_2\dots x_k,$ and $\pb{x_{\ell},x_{n-k+1}\dots
  x_n}_k^{(n)}=kx_{\ell}x_{n-k+1}\dots x_n$, for any $\ell$ with $k<\ell\leqs n-k$.
\end{proof}
We will also make use of the involution $\psi:\bbR^n\to\bbR^n$, defined by
\begin{equation}\label{eq:psi_def}
  \psi(x_1,x_2,\dots,x_n):=(x_n,\dots,x_2,x_1)\;.
\end{equation}
Clearly it is, for any $k$ with $0\leqslant k<n$, an anti-Poisson map from $(\bbR^n,\pin nk)$ to itself.

\section{Definition of the systems $\LV(n,k)$ and their first integrals} \label{sec:definition_of_systems_and_integrals}

We now introduce the Lotka-Volterra lattices which will be studied in this paper. Fix $n$ and $k$ with $0\leqs
k<n$. Let us recall from (\ref{eq:A_k_def_2}) that $A_k$ denotes the skew-symmetric $n\times n$ Toeplitz matrix,
whose first row is given by $(0,1,1,\dots,1, -1,-1,\dots,-1)$, with $-1$ appearing $k$ times. The corresponding
Lotka-Volterra system is given by
\begin{equation}\label{eq:LV}
  \dot x_i=\sum_{j=1}^n (A_k)_{i,j} x_ix_j\;.
\end{equation}%
We will denote this system by $\LV(n,k)$. It is a Hamiltonian system, with Hamiltonian $H:=x_1+x_2+\dots+x_n$ and
Poisson structure $\pi_k$. As pointed out in the introduction, the Liouville and superintegrability of $\LV(n,0)$
have been shown recently by van der Kamp et al., \cite{KKQTV} while the Liouville integrability of $\LV(2k+1,k)$
has been established by Bogoyavlenskij \cite{Bog2} and Itoh \cite{itoh1,itoh2}. The systems which will be
considered here interpolate between these two integrable systems in the following sense. Consider $\LV(n,k)$, where
$n>2k+1$. On the one hand, setting the last $k$ coordinates of $\bbR^n$ equal to zero, we arrive at the reduced
Hamiltonian system $\LV(n-k,0)$, which is of the type studied in \cite{KKQTV}. On the other hand, $\LV(n,k)$ can be
obtained by reduction from the Bogoyavlenskij-Itoh system $\LV(2n-2k-1,n-k-1)$ by setting the last $n-2k-1$
coordinates of $\bbR^{2n-2k-1}$ equal to zero. In what follows, it is these systems $\LV(n,k)$, with $n>2k+1$,
which we will analyze from the integrable point of view.

For future reference we first give a Lax equation for (\ref{eq:LV}). For the special case
of $\LV(2k+1,k)$ the following Lax equation, with spectral parameter~$\l$, was provided by Bogoyavlenskij in
\cite{Bog1}:
\begin{equation}\label{eq:bogo_lax}
  (X+\lambda M)^\cdot=[X+\lambda M,B-\lambda M^{k+1}]
\end{equation}
where for $1\leqslant i,j\leqslant 2k+1$ the $(i,j)$-th entry of the matrices $X,\,M$ and $B$ is respectively given
by
\begin{equation}\label{eq:bogo_lax_entries}
  X_{i,j}:=\delta_{i,j+k}x_i\;,\quad M_{i,j}:=\delta_{i+1,j}\;,\quad
  B_{i,j}:=b_i:=-\delta_{i,j}(x_i+x_{i+1}+\cdots+x_{i+k})\;.
\end{equation}%
In the right hand side of these formulas, all indices are taken modulo $2k+1$ so that, for example, $M_{2k+1,1}=1$.
To check that (\ref{eq:bogo_lax}) is equivalent to (\ref{eq:LV}) (with $n=2k+1$) it is sufficient to check that
(\ref{eq:LV}) is equivalent with $\dot X=[X,B]$ and (since $M$ is constant) that $[M,B]-[X,M^{k+1}]=0$.  For the
latter, one finds at once from (\ref{eq:bogo_lax_entries}) that
\begin{equation*}
  ([M,B]-[X,M^{k+1}])_{i,j}=\delta_{i+1,j}(b_j-b_i-x_i+x_{j+k})=0\;.
\end{equation*}%
Also, since $B$ is a diagonal matrix, $[X,B]_{i,j}=X_{i,j}(b_j-b_i)$, with non-zero entries only when $j=i-k$; for
these entries, one has from the Lax equation
\begin{equation*}
  \dot x_i=\dot X_{i,i-k}=[X,B]_{i,i-k}=x_i(b_{i-k}-b_i)\;,
\end{equation*}%
which is the right hand side of (\ref{eq:LV}) (recall that $n=2k+1$). The Lax equation for
the general case ($n>2k+1$) is obtained from this Lax equation by substituting $0$ for the last variables.

\subsection{The rational first integrals}\label{subsec:rational_integrals}
We will first construct a set of rational first integrals for $\LV(n,k)$, where $n>2k+2$. To do this, we will use
the map $\phi_k$, defined in Proposition~\ref{prop:poisson_map}: we construct $n-2k-2$ rational functions on
$\bbR^n$ by pulling back (using $\phi_k$) the $n-2k-2$ independent rational first integrals of $\LV(n-2k,0)$
(except the Hamiltonian), which were constructed in~\cite{KKQTV}. We will then show that this yields $n-2k-2$
independent first integrals of $\LV(n,k)$.

We first recall the explicit formulas for the rational first integrals that were introduced in
\cite{KKQTV}. Setting $m:=n-2k$ and $r:=\left[\frac{m+1}2\right]$ and denoting the coordinates on $\bbR^{m}$ by
$y_1,\dots,y_m$, the first set of rational first integrals of $\LV(n-2k,0)$ (roughly the first half) is given for
$1\leqs \ell\leqs r$ by
\begin{equation}\label{E:integrals}
  F_\ell:= \left\{ \begin{array}{ll}
    \left(y_1+y_2+\cdots+y_{2\ell-1}\right)\ds\frac{y_{2\ell+1}y_{2\ell+3}\ldots y_{m}}{y_{2\ell}y_{2\ell+2}
        \ldots y_{m-1}}&\mbox{ when}\  m\mbox{ is odd},\\
    \\
    \left(y_1+y_2+\cdots+y_{2\ell}\right)\ds\frac{y_{2\ell+2}y_{2\ell+4}\ldots y_m}{y_{2\ell+1} y_{2\ell+3}
          \ldots y_{m-1}}&\mbox{ when} \  m\mbox{ is even}.
    \end{array} \right.
\end{equation}
The other rational first integrals are obtained by using the involution $\psi$ (see (\ref{eq:psi_def})): set
$G_\ell:=\psi^* F_\ell$ for $1\leqs \ell\leqs r$. It leads to the following $m-1$ different (in fact, functionally
independent) functions:
\begin{align}
  F_1=G_1,F_2,\dots,F_{r-1},G_2,\dots,G_{r-1},F_r=G_{r},\mbox{ when $m$ is odd,} \label{eq:fun_odd}\\
  F_1,\dots,F_{r-1},G_1,\dots,G_{r-1},F_r=G_{r},\mbox{ when $m$ is even.} \label{eq:fun_even}
 \end{align}
We denote the pull-backs via $\phi_k$ of these functions (in that order) by $H^{(n,k)}_1,H^{(n,k)}_2,$
$\dots,H^{(n,k)}_{m-1}$. In formulas, this means that
\begin{equation}\label{eq:pullbacks}
  H^{(n,k)}_\ell:=\phi_k^*  H^{(n-2k,0)}_\ell\;,\mbox{ for $\ell=1,\dots,m-1\;,$}
\end{equation}%
where $H^{(n-2k,0)}_1,\dots,H^{(n-2k,0)}_{m-1}$ stand for the functions in (\ref{eq:fun_odd}) or
(\ref{eq:fun_even}). In what follows, we will not consider the last function, to wit
$H^{(n,k)}_{m-1}=\phi^*_kF_r=\phi^*_kG_r$; in fact, $F_r$ is the Hamiltonian of $\LV(m,0)$, $F_r=y_1+\dots,y_m$,
and so $\phi^*_kF_r$ is a \emph{polynomial} first integral which we will recover in a different way in the next
section, together with the other polynomial first integrals.

For example, when $n$ is odd, the fact that $\phi_k^*y_i=x_1x_2\dots x_kx_{i+k}x_{n-k+1}\dots x_{n-1}x_n$ implies
for $\ell=1,\dots,r-1=\frac{n-1}2-k$ that
\begin{eqnarray}\label{eq:H_odd}
  H^{(n,k)}_\ell&=&x_1x_2\dots x_k\left(x_{k+1}+x_{k+2}+\cdots+x_{k+2\ell-1}\right)
    \ds\frac{x_{k+2\ell+1}x_{k+2\ell+3}\ldots x_{n-k}}{x_{k+2\ell}x_{k+2\ell+2}
      \ldots x_{n-k-1}}x_{n-k+1}\dots x_{n-1}x_n\;,\nonumber\\
    &=&\hat H^{(n,k)}_\ell\left(x_{k+1}+x_{k+2}+\cdots+x_{k+2\ell-1}\right)\;,
\end{eqnarray}%
where we have introduced in the last line a notation\footnote{For $n$ even, $\hat H^{(n,k)}_\ell$ is defined in the
same way, to the effect that $H^{(n,k)}_\ell=\hat H^{(n,k)}_\ell\left(x_{k+1}+x_{k+2}+\cdots+
x_{k+2\ell}\right).$}, which will turn out to be very useful. The functions $H^{(n,k)}_\ell$ and $\hat
H^{(n,k)}_\ell$, with $\ell=r,\dots,n-2k-2$ can be obtained by applying $\psi^*$ to these functions, because
$\phi_k$ and $\psi$ commute.

We will now show that the functions $H^{(n,k)}_\ell$ are first integrals of $\LV(n,k)$. To do this, we will use the
following lemma:
\begin{lemma}\label{lma:xs}
  Let $\ell=1,\dots,n-2k-2$ and let $j$ denote an index which is present in the sum which appears in
  $H^{(n,k)}_\ell$ (see (\ref{eq:H_odd})).
  \begin{enumerate}
    \item[(1)] If the variable $x_s$ appears in $x_j\hat H^{(n,k)}_\ell$ then $\pb{x_s,x_j\hat
      H^{(n,k)}_\ell}_k^{(n)}=0$\;;
    \item[(2)] If the variable $x_s$ does not appear in $\hat H^{(n,k)}_\ell$ then $\pb{x_s,\hat
      H^{(n,k)}_\ell}_k^{(n)}=0$\;.
  \end{enumerate}
\end{lemma}
\begin{proof}
We give the proof for $n$ odd. Using the involution $\psi$ if necessary, we may suppose that
$1\leqslant\ell\leqslant \frac{n-1}2-k$. Then $j$ satisfies $k<j< k+2\ell$. In order to prove (1), let us first
suppose that $1\leqs s\leqs k$. Then it follows from (\ref{eq:A_k_def_2}) and (\ref{eq:pik_def_2}) that
\begin{align}\label{eq:brack_lemma}
  &\pb{x_s,x_1x_2\dots x_k}=(k-2s+1)x_sx_1x_2\dots x_k\;, \nonumber\\
  &\pb{x_s,x_{n-k+1}\dots x_{n-1} x_n}=(2s-k-2)x_sx_{n-k+1}\dots x_{n-1} x_n\;,\\
  &\pb{x_s,x_j}=x_sx_j\;,\quad\hbox{ and }\quad \pb{x_s,x_{t+1}/x_t}=0\hbox{ if } t=k+2\ell,\dots,n-k-1\;.\nonumber
\end{align}
It follows from these formulas that $\pb{x_s,x_j\hat H^{(n,k)}_\ell}=0$ when $1\leqs s\leqs k$. The proof for $s$
satisfying $n-k+1\leqs s\leqs n$ is essentially the same. When $s=j$, the above formulas (\ref{eq:brack_lemma}) get
replaced by
\begin{align}\label{eq:brack_lemma_2}
  &\pb{x_s,x_1x_2\dots x_k}=-kx_sx_1x_2\dots x_k\;, \nonumber\\
  &\pb{x_s,x_{n-k+1}\dots x_{n-1} x_n}=kx_sx_{n-k+1}\dots x_{n-1} x_n\;,\\
  &\pb{x_s,x_j}=0\;,\quad\hbox{ and }\quad \pb{x_s,x_{t+1}/x_t}=0\hbox{ if } t=k+2\ell,\dots,n-k-1\;,\nonumber
\end{align}
and one arrives at the same conclusion. Finally, when $k+2\ell\leqs s\leqs n-k$ the first two formulas and the last
formula in (\ref{eq:brack_lemma_2}) are still valid, the third one gets replaced by $\pb{x_s,x_j}=-x_sx_j$, and the
last one gets replaced, depending on whether $s$ is even or odd (in that order) by
\begin{equation*}
  \pb{x_s,x_{s+1}/x_s}=x_sx_{s+1}/x_s\; \quad\hbox{ or }\quad
  \pb{x_s,x_{s}/x_{s-1}}=x_s^2/x_{s-1}\;.
\end{equation*}%
In either case, it follows again that $\pb{x_s,x_j\hat H^{(n,k)}_\ell}=0$. This finishes the proof of item (1).
Item (2) is an immediate consequence of item (1) because if $x_s$ does not appear in~$\hat H^{(n,k)}_\ell$ then
(still assuming that $1\leqslant\ell\leqslant \frac{n-1}2+k$) $k<s<k+2\ell$ and so $0=\pb{x_s,x_s\hat
  H^{(n,k)}_\ell}=x_s\pb{x_s,\hat H^{(n,k)}_\ell}$.
\end{proof}
\begin{prop}
  For any $k$ such that $n-2k-2>0$, the rational functions $H^{(n,k)}_\ell$ with $\ell=1,\dots,n-2k-2$ are first
  integrals of (\ref{eq:LV}).
\end{prop}
\begin{proof}
Again we give the proof only for $m$ odd. Since (\ref{eq:LV}) is the Hamiltonian vector field associated to
$H=\sum_{i=1}^nx_i$, it suffices to prove that $H^{(n,k)}_\ell$ and $H$ are in involution. This is shown in the
following computation, where we use item (1) of Lemma \ref{lma:xs} in the second step and item (2) in the fourth
step:
\begin{eqnarray*}
  \pb{H^{(n,k)}_\ell,H} &=& \pb{\hat H^{(n,k)}_\ell\sum_{i=k+1}^{k+2\ell-1}x_i,\sum_{i=1}^{n}x_i}=
  \pb{\hat H^{(n,k)}_\ell\sum_{i=k+1}^{k+2\ell-1}x_i,\sum_{i=k+1}^{k+2\ell-1}x_i}\\
  &=&\pb{\hat H^{(n,k)}_\ell,\sum_{i=k+1}^{k+2\ell-1}x_i}\sum_{i=k+1}^{k+2\ell-1}x_i=0\;.
\end{eqnarray*}%
\end{proof}

\subsection{The polynomial first integrals}\label{subsec:polynomial_integrals}
We will now construct $k$ independent polynomial first integrals for $\LV(n,k)$, besides the Hamiltonian $H$. We do
this by using the polynomial invariants which Bogoyavlenskij constructed for $\LV(2k+1,k)$ from the Lax equation
(\ref{eq:bogo_lax}). The characteristic polynomial of $X+\lambda M$ has the form
\begin{equation}\label{eq:bogo_char_poly}
  \det(X+\lambda M-\mu\Id)=\l^{2k+1}-\mu^{2k+1}+\sum_{i=0}^kK_i\l^{k-i}\mu^{k-i}\;,
\end{equation}%
where, by homogeneity, each $K_i$ is a homogeneous polynomial (in $x_1,\dots,x_{2k+1}$) of degree $2i+1$. One has
$K_0=x_1+x_2+\dots+x_{2k+1}=H$, the Hamiltonian, and $K_k=x_1x_2\dots x_{2k+1}$, which is a Casimir of
$\LV(2k+1,k)$. Being a coefficient of the characteristic polynomial of the Lax operator $X+\l M$, each one of the
$K_i$ is a first integral of $\LV(2k+1,k)$. In view of Proposition \ref{prp:inclusion}, the restrictions of these
integrals $K_i$ to $\LV(2k,k-1),\ \LV(2k-1,k-2),\dots,\LV(k+1,0)$ lead to first integrals for these systems, but
these restrictions may be trivial (zero). In order to find simpler formulas for these restrictions and to see when
they are zero, we give a combinatorial description of the polynomials~$K_i$; the description that we give is a
matricial reformulation of Itoh's original combinatorial description, given in \cite{itoh1}.

Fix $n$ and $k$ with $1<2k+1\leqs n$ and consider the matrix $A_k:=A_k^{(n)}$ defined in (\ref{eq:A_k_def_2}). Fix
$i\in\set{1,\dots,k}$ and let $\um=(m_1,m_2,\dots,m_{2i+1})$ be an $2i+1$-tuple of integers, satisfying $1\leqs
m_1<m_2<\cdots<m_{2i+1}\leqs n$. We view them as indices of the rows and columns of $A_k$: we denote by $B_{\um}$
the square submatrix of $A_k$ of size $2i+1$, corresponding to rows and columns $m_1,m_2,\dots,m_{2i+1}$ of~$A_k$,
so that
\begin{equation}\label{eq:B_def}
  (B_{\um})_{s,t}=(A_k)_{m_s,m_t}\;, \hbox{ for } s,t=1,\dots,2i+1\;.
\end{equation}
Let
\begin{equation}\label{eq:S_def}
  \cS{n,k}{i}:=\set{\um\mid B_{\um}=A_i^{(2i+1)}}\;.
\end{equation}
As was pointed out by Bogoyavlenskij, the polynomials $K_i$ which appear in the characteristic polynomial
(\ref{eq:bogo_char_poly}) can be written as
\begin{equation}\label{eq:k_i_itoh}
  K_i=\sum_{\um\in \cS{2k+1,k}{i}} x_{m_1}x_{m_2}\dots x_{m_i}\dots x_{m_{2i+1}}\;.
\end{equation}%
For example, $\cS{2k+1,k}{0}=\set{1,2,\dots,2k+1}$ and $\cS{2k+1,k}{k}=\set{(1,2,\dots,2k+1)}$, so that
$K_0=x_1+x_2+\dots+x_{2k+1}$ and $K_k=x_1x_2\dots x_{2k+1}$, as above.

We use the latter description to give a combinatorial formula for the restrictions of the integrals $K_i$, obtained
by setting the last few variables equal to zero. Suppose that we put the last $\ell\leqs k$ variables
$x_{2k-\ell+2},\ x_{2k-\ell+3},\dots,x_{2k+1}$ equal to zero, which leads us by reduction to
$\LV(2k+1-\ell,k-\ell)$. Consider a first integral~$K_i$ of $\LV(2k+1,k)$, as defined in (\ref{eq:k_i_itoh}). Since
the restriction of $K_i$ to $\LV(2k+1-\ell,k-\ell)$ is obtained by replacing the last $\ell$ variables
$x_{2k-\ell+2},\ x_{2k-\ell+3},\dots,x_{2k+1}$ by~$0$, the sum in (\ref{eq:k_i_itoh}) can be restricted to the
$(2i+1)$-tuplets $\um=(m_1,m_2,\dots,m_{2i+1})$, with $m_{2i+1}\leqs 2k-\ell+1$; thus, we can view these integers now
as the rows and columns of a submatrix of $A_k^{(2k+1)}$ obtained from it by removing from it its last $\ell$ rows
and columns, i.e., as the rows and columns of $A_{k-\ell}^{(2k-\ell+1)}$. For future reference, we state this in the
following proposition.
\begin{prop}\label{prp:itoh_rest}
  Suppose that $1<2k+1<n$. For $i=0,\dots,k$ the polynomial $K_i^{(n,k)}$, defined by
  \begin{equation}\label{eq:k_i_itoh_rest}
    K_i^{(n,k)}:=\sum_{\um\in \cS{n,k}{i}} x_{m_1}x_{m_2}\dots x_{m_i}\dots x_{m_{2i+1}}\;
  \end{equation}%
  is a first integral of $\LV(n,k)$.
\end{prop}
Notice that $K_i^{(n,k)}$ is homogeneous and has degree $2i+1$. Notice also that when $i>k$ the set $\cS{n,k}{i}$
is empty; said differently, when $i>k$ the restriction of $K_i$ to $\LV(n,k)$ is zero. We will see below that the
polynomials $K_0^{(n,k)}=H,K_1^{(n,k)},\dots, K_k^{(n,k)}$ are actually functionally independent, in particular
they are not trivial.

Since the polynomials $K_i^{(n,k)}$ are defined in terms of the sets $\cS{n,k}{i}$, we need a characterization of
the elements of the latter sets. It is given in the following proposition.
\begin{prop}\label{lma:S}
  Suppose that $n\geqs2k+1$ and let $\um=(m_1,\dots,m_{2i+1})$ be a strictly ordered $2i+1$-tuplet of elements of
  $\set{1,2,\dots,n}$. Then $\um \in \cS{n,k}{i}$ if and only if the following conditions are satisfied:
  \begin{enumerate}
    \item[(1)] $m_{i+s}<m_s+n-k\leqs m_{i+s+1}$ for $s=1,\dots,i$;
    \item[(2)] $m_{2i+1}<m_{i+1}+n-k$.
  \end{enumerate}
\end{prop}
\begin{proof}
Suppose that $\um=(m_1,\dots,m_{2i+1})$ with $1\leqs m_1<m_2<\dots<m_{2i+1}\leqs n$. In
view of the definitions (\ref{eq:B_def}) and (\ref{eq:S_def}) of $B_{\um}$ and $\cS{n,k}{i}$, we have that
$\um\in \cS{n,k}{i}$ if and only if
\begin{equation*}
  \(A_k^{(n)}\)_{m_s,m_t}=\(A_i^{(2i+1)}\)_{s,t}\;, \quad\mbox{ for }\quad 1\leqs s,t\leqs 2i+1\;.
\end{equation*}%
Using (\ref{eq:A_k_def_2}) this condition can be translated into
\begin{align*}
  &n+m_s>k+m_t\quad\mbox{ when }\quad i+1+s>t\;,\\
  &n+m_s\leqs k+m_t\quad\mbox{ when }\quad i+1+s\leqs t\;,
\end{align*}
which is equivalent to
\begin{align}
  &n+m_s>k+m_t\quad\mbox{ when }\quad i+s=t\;,\label{eq:S_ineqs_2a}\\
  &n+m_s\leqs k+m_t\quad\mbox{ when }\quad i+1+s=t\;.\label{eq:S_ineqs_2b}
\end{align}
The latter equivalence is a direct consequence of the fact that $\um$ is strictly increasing, i.e., $m_s<m_t$ when
$s<t$. The conditions (\ref{eq:S_ineqs_2a}) and (\ref{eq:S_ineqs_2b}) yield for $s=1,\dots,i$ precisely item (1),
while item (2) is obtained by taking $s=i+1$ in (\ref{eq:S_ineqs_2a}), which is the only remaining possible value
for $s$ in (\ref{eq:S_ineqs_2a}) and (\ref{eq:S_ineqs_2b}) such that $1\leqs i,j\leqs 2i+1$.
\end{proof}
We list a few properties of the elements $\um$ of $\cS{n,k}{i}$ which are direct consequences of Proposition
\ref{lma:S}.
\begin{cor}\label{cor:S}
  Let  $n$ and $k$ be integers with $n\geqs2k+1$. Suppose that $\um=(m_1,m_2,\dots,m_{2i+1})\in\cS{n,k}{i}$ and denote
  by $\um'$ the vector $\um$ with its middle entry $m_{i+1}$ replaced by $m'_{i+1}$.
  \begin{enumerate}
    \item[(1)] $m_i\leqs k$;
    \item[(2)] $m_{i+2}> n-k\geqs k+1$;
    \item[(3)] $\um'\in\cS{n,k}{i}$ if and only if $m_{2i+1}-n+k<m'_{i+1}<m_1+n-k$;
    \item[(4)] $\um'\in\cS{n,k}{i}$ when  $k<m'_{i+1}<n-k+1$.
  \end{enumerate}
\end{cor}
\begin{proof}
If we take $s:=i$ in Proposition (\ref{lma:S}) (1), we find $m_i\leqs m_{2i+1}-n+k\leqs k$, which is item (1). The
first part of item (2) is obtained similarly by taking $s:=1$ in the same inequality; the second part of item (2)
follows from $n\geqs2k+1$. When we replace $m_{i+1}$ by $m'_{i+1}$ the only inequalities in Proposition (\ref{lma:S})
which get affected are (1) with $s:=1$ and (2); they become precisely the two inequalities in item (3). If
$k<m'_{i+1}<n-k+1$ then $m_{2i+1}-n+k<m'_{i+1}<m_1+n-k$, so item (4) is a consequence of item (3).
\end{proof}
One more property of the elements $\um$ of $\cS{n,k}{i}$ is given in the following example.
\begin{example}
\label{example_bogo_case_S_i,j_same_cardinality}
Let $n\geq 2k+1$ and suppose that $\um=(m_1,m_2,\ldots,m_{2i+1})\in\cS{n,k}{i}$ with $m_{2i+1}<n$.
Then from the conditions given in Proposition \ref{lma:S} it easily follows that
$\um'=(m_1+1,m_2+1,\ldots,m_{2i+1}+1)$ also belongs to $\cS{n,k}{i}$.
In the case $n=2k+1$, if $\um=(m_1,m_2,\ldots,m_{2i+1})\in\cS{n,k}{i}$ with $m_{2i+1}=n$ then
$\um':=(m_1',m_2',\ldots,m_{2i+1}')=(1,m_1+1,m_2+1,\ldots,m_{2i}+1)$ belongs to $\cS{n,k}{i}$.
Indeed the only condition needed to be checked is the $m_{i+1}'<m_1'+n-k\leqs m_{i+2}'$ which
translates to $m_{i}<k+1\leqs m_{i+1}$. This follows from Corollary \ref{cor:S} (items (1) and (4)).
This shows that the cyclic group of $n$ elements acting on $\mathbb R^n$ by permuting the variables,
leaves the first integrals $K_i^{(2k+1,k)}$ invariant.
\end{example}

\section{Independence of the first integrals}
We have constructed in the previous section $n-k-1$ (polynomial and rational) first integrals for
$\LV(n,k)$, where $n>2k+1$. We prove now the following result, concerning the independence of these first
integrals.
\begin{prop}\label{prp:independence}
  The $n-k-1$ first integrals $H^{(n,k)}_1,\dots,H^{(n,k)}_{n-k-2},$ $K_0^{(n,k)},$ $K_1^{(n,k)},$ $\dots,$ $K_k^{(n,k)}$ of
  $\LV(n,k)$ are functionally independent.
\end{prop}
The proof of this proposition is quite long and technical; it will take up this whole section and can be
skipped on a first reading, as the rest of the paper only depends on the statement of the above proposition, and
not on its proof.

We only need to show that the differentials of the above first integrals are independent at some point of $\bbR^n$:
since these functions are polynomial or rational, their differentials will then be independent on an open dense
subset of $\bbR^n$, proving their functional independence. To do this, we show that the Jacobian matrix of these
first integrals with respect to the $n-k$ variables $x_1,\dots,x_{n-k}$ is of maximal rank ($n-k-1$) at the point
${\bf~1}=(1,1,\dots,1)$ of $\bbR^n$.  More precisely, we show that there exist constants $p_\ell$ and $q_i$ (with
$\ell=1,2,\ldots,n-2k-2$ and $i=1,2,\ldots,k$) such that the Jacobian matrix at $\bf 1$ of the following functions
(which are the above first integrals, shifted by a multiple of the Hamiltonian $H=K_0^{(n,k)}$),
\begin{gather}
  H^{(n,k)}_\ell-p_\ell H \text{ for } \ell=1,2,\ldots,n-2k-2\;,\nonumber\\
  H=K_0^{(n,k)}\;,\label{eq:integrals_for_jac}\\
  K^{(n,k)}_i-q_i H \text{ for }i=1,2,\ldots,k\;,\nonumber
\end{gather}
has the following form
\begin{equation}\label{eq:jac_at_1}
\begin{pmatrix}
  {\bf 0}_{n-2k-2,k}&\Phi_{n-2k-2,n-2k}\\
  {\bf 1}_{1,k}&{\bf 1}_{1,n-2k}\\
  \Lambda_{k,k}&{\bf 0}_{k,n-2k}
\end{pmatrix}\;,
\end{equation}
and is of maximal rank; in this block matrix, the subscripts denote the dimension of the different blocks. Also,
the matrices ${\bf 1}$ and $\bf 0$ have all entries equal to 1, respectively to 0.

We first prove the existence of the constants $p_\ell$ and $q_i$. When $n$ is odd, it follows from (\ref{eq:H_odd})
that
\begin{equation}\label{proposition_partials_of_H_i}
  \frac{\partial H_\ell^{(n,k)}}{\partial x_j}({\bf 1})=2\ell-1\;,\quad\hbox{ for }\quad j=1,\dots,k\;,
\end{equation}
and so, since $\pp H{x_j}=1$, it suffices to define $p_\ell:=2 \ell-1$ for $\ell=1,2,\ldots,n-2k-2$ to obtain the
upper left block of zeroes in (\ref{eq:jac_at_1}). Similarly, when $n$ is even, $p_\ell:=2\ell$ does the job. Also,
it follows from (4) in Corollary \ref{cor:S} that the number of monomials in $K_i^{(n,k)}$ containing $x_j$ does
not depend on $j$ when $k<j<n-k+1$; their number is the number $q_i$ needed to obtain the lower right block of
zeroes in (\ref{eq:jac_at_1}) since all these monomials have a coefficient $1$, and so $\pp{K_i^{(n,k)}}{x_j}=q_i$.

It remains to be shown that the matrix (\ref{eq:jac_at_1}) has maximal rank. It is shown in~\cite{KKQTV} that the
Jacobian matrix
$$
  \frac{\partial (H_1^{(n-2k,0)},\ldots,H_{n-2k-1}^{(n-2k,0)})}{\partial (y_1,y_2,\ldots,y_{n-2k})}({\bf 1})
$$
is of full rank ($n-2k-1$). Since $H_\ell^{(n,k)}=H_\ell^{(n-2k,0)}\circ\phi_k$, we have
\begin{equation*}
  \pp{H_\ell^{(n,k)}}{x_j}({\bf 1})=  \pp{H_\ell^{(n-2k,0)}}{y_{j-k}}({\bf 1})\;,
  \quad\hbox{ for }\quad j=k+1,\dots,n-k\;,
\end{equation*}%
and so the rank of the Jacobian matrix
\begin{equation}\label{eq:jac1}
  \frac{\partial (H_1^{(n,k)},\ldots,H_{n-2k-1}^{(n,k)})}{\partial (x_{k+1},x_{k+2},\ldots,x_{n-k})}({\bf 1})
\end{equation}
is also maximal. Now $\Phi_{n-2k-2,n-2k}$ is given by
\begin{equation*}
    \Phi_{n-2k-2,n-2k}=\frac{\partial (H_1^{(n,k)}-p_1H,\ldots,H_{n-2k-2}^{(n,k)}-p_{n-2k-2}H)}
  {\partial (x_{k+1},x_{k+2},\ldots,x_{n-k})}({\bf 1})\;,
\end{equation*}%
and all entries below it (in (\ref{eq:jac_at_1})) are equal to 1. It follows that the matrices
(\ref{eq:jac1}) and $\begin{pmatrix}\Phi_{n-2k-2,n-2k}\\ {\bf 1}_{1,n-2k}\end{pmatrix}$ coincide, up to some row
operations; in particular, they have maximal rank.

We still need to show that $\Lambda_{k,k}$ also has maximal rank. For the proof, we need several notations and
relations which are of combinatorial nature. We first introduce the notation that we will use. First, we denote by
$\cK$ or $\cK^{(n,k)}$ the Jacobian matrix
\begin{equation}\label{eq:jaco}
  \cK^{(n,k)}:=\frac{\p(K_1^{(n,k)},\dots,K_k^{(n,k)})}{\p(x_1,\dots,x_k)}({\bf 1})\;,\hbox{ so that
  }\cK^{(n,k)}_{i,j}=\pp{K_i^{(n,k)}}{x_j}({\bf1})\;.
\end{equation}%
For $\um=(m_1,m_2,\dots,m_{2i+1})$ we denote by $\hum$ the vector $\um$ with its middle element removed,
\begin{equation*}
  \hum=(m_1,m_2,\dots,m_i,m_{i+2},\dots,m_{2i+1})\;.
\end{equation*}%
We deduce from the definition (\ref{eq:S_def}) of $\cS{n,k}{i}$ three related sets
\begin{gather*}
  \cS{n,k}{i,j}:=\set{\um\in\cS{n,k}{i} \mid j\in\set{m_1,\dots,m_{i+1}}}\;,\\
  \hcS{n,k}{i}:=\set{\hum\mid\um\in\cS{n,k}{i}}\;,\quad
  \hcS{n,k}{i,j}:=\set{\hum\in\hcS{n,k}{i} \mid j\in\set{m_1,\dots,m_i}}\;,
\end{gather*}%
where $j=1,\dots,k$. Finally, we put
\begin{equation}\label{eq:s_def}
  \s^{(k)}_{i,j}:=\#\hcS{2k+1,k}{i,j}\;.
\end{equation}%
In the following proposition we relate the entries of the matrix $\Lambda_{k,k}$ with those of $\cK$ and with the
numbers $\s^{(k)}_{i,j}$ for which we give a formula; combining these relations, we will prove that $\Lambda_{k,k}$
is of maximal rank.
\begin{prop}\label{prp:indep_tech}
  Let $n,k$ be such that $n>2k+1$ and let $i,j\in\set{1,2,\dots,k}$.
  \begin{enumerate}
    \item[(1)] The entries of $\cK$ are given by
      \begin{equation*}
        \cK^{(n,k)}_{i,j}=\#\cS{n,k}{i,j}\;;
      \end{equation*}%
    \item[(2)] The entries of $\Lambda_{k,k}$ and of $\cK$ are related by
      \begin{equation*}
        (\Lambda_{k,k})_{i,j}=\cK^{(n,k)}_{i,j}-q_i\;;
      \end{equation*}%
    \item[(3)] The assignment $\um\mapsto \um'$, where $m'_s:=m_s$ for $s=1,\dots,i+1$ and $m'_s:=m_s+1$ for
      $s=i+2,\dots,2i+1$ defines a map $\rho:\cS{n-1,k}i\to\cS{n,k}i$;
    \item[(4)] The map $\rho$ induces bijections
      $\hat\rho:\hcS{n-1,k}i\to\hcS{n,k}i$ and $\hat\rho_j:\hcS{n-1,k}{i,j}\to\hcS{n,k}{i,j}$ for $j=1,\dots,k$. In
      particular,
      \begin{equation*}
        \#\hcS{n,k}{i,j}=\s^{(k)}_{i,j}\;\quad\hbox{ for all } n\geq2k+1\;;
      \end{equation*}%
    \item[(5)] The entries of $\cK$ and the numbers $\s^{(k)}_{i,j}$ are related by
      \begin{equation*}
        \cK^{(n,k)}_{i,j}-\cK^{(n-1,k)}_{i,j}=\s^{(k)}_{i,j}\;;
      \end{equation*}%
    \item[(6)] For $j<k$,
      \begin{equation}\label{eq:rec_for_sigma}
        \s^{(k)}_{i,j}-\s^{(k)}_{i,j+1}=\frac1{(2i-2)!}\prod_{s=1-i}^{i-2}(2j-k+s)\;,
      \end{equation}%
      where the right hand side is, by definition, equal to 1 when $i=1$.
  \end{enumerate}
\end{prop}
\begin{proof}
From the definition (\ref{eq:k_i_itoh_rest}) of $K^{(n,k)}_i$, combined with (\ref{eq:jaco}), we find that
\begin{equation*}
  \cK^{(n,k)}_{i,j}=\#\set{\um\in\cS{n,k}i\mid j\in\set{m_1,m_2,\dots,m_{2i+1}}}\;.
\end{equation*}%
In order to derive (1) from it suffices to use the inequalities $j\leqs k<m_{i+2}$ (see item (2) in Corollary
\ref{cor:S} for the second inequality).  In view of (\ref{eq:integrals_for_jac}) and (\ref{eq:jac_at_1}), the
matrix $\Lambda_{k,k}$ is by definition given by
\begin{equation*}
  \Lambda_{k,k}=\frac{\p(K_1^{(n,k)}-q_1H,\ldots,K_k^{(n,k)}-q_kH)}{\p (x_{1},x_{2},\ldots,x_{k})}({\bf 1})\;,
\end{equation*}
from which item (2) follows. In order to prove item (3), we need to show that when $\um$ satisfies the two
conditions of Proposition (\ref{lma:S}), then $\um'$, as defined in item (3), also verifies them (with $n$ replaced
by $n+1$). In these conditions, every term is augmented by 1, proving their validity, except for condition (1) with
$i=1$, where one has to check that $m_{i+1}<m_1+n-1-k$ implies that $m_{i+1}<m_1+n-k$, but this is trivial. This
proves (3).

Since the map $\um\mapsto\hum$ amounts to removing the middle entry of its argument and since, by definition,
$\hcS{n,k}i$ is the image of this map, $\rho$ induces a map $\hat\rho:\hcS{n-1,k}i\to\hcS{n,k}i$, which is by
construction injective; explicitly it is given by $(m_1,\dots,m_{i},m_{i+2},\dots,m_{2i+1})\mapsto
(m_1,\dots,m_{i},m_{i+2}+1,\dots,m_{2i+1}+1)$. To show that $\hat\rho$ is surjective, choose as representative
$\um'$ for a given $\hum'\in\hcS{n,k}i$ the one for which $m'_{i+1}=k+1$; this yields indeed an element $\um'$ of
$\cS{n,k}i$, according to item (4) in Corollary \ref{cor:S}.  Thanks to this choice, $\um'=\rho(\um)$ with
$\um\in\cS{n-1,k}i$, by the same use of Proposition (\ref{lma:S}) as above: the exceptional case of (1) with $i=1$
now amounts to checking that $m_{i+1}=k+1<m_1+n-1-k$, which is fine since $2k+1<n$. Then $\hat\rho(\hum)=\hum'$, so
that $\hat\rho$ is surjective, hence bijective. For future reference, notice that if one can pick a representative
$\um'$ in $\cS{n-1,k}i$ for $\hum'$ with $m'_{i+1}<k+1$, then this representative is also in the image of
$\rho$. If, in the bijection~$\hat\rho$, $j$ appears as one of the (first $i$) indices of $\um$, the same will be
true for~$\um'$, and vice versa. Therefore, $\hat\rho_j$ is also bijective, for $j=1,\dots,i$. From it and from the
definition (\ref{eq:s_def}) of $\s^{(k)}_{i,j}$, we get
\begin{equation*}
  \s^{(k)}_{i,j}=\#\hcS{2k+1,k}{i,j}=\#\hcS{n,k}{i,j}\;,\quad\hbox{ for all } n\geq2k+1\;.
\end{equation*}%
This proves the different claims in (4). We next prove (5). In view of items (1) and~(4) we need to show that
\begin{equation}\label{eq:card_1}
  \#\cS{n,k}{i,j}=\#\cS{n-1,k}{i,j}+\#\hcS{n-1,k}{i,j}\;.
\end{equation}%
Let us denote for all $n\geq2k+1$ and for $j=1,\dots,k$ by $\cS{n,k}{i,j=m_{i+1}}$ and $\cS{n,k}{i,j<m_{i+1}}$ the subsets of
$\cS{n,k}{i,j}$ consisting of those $\um$ for which $j=m_{i+1}$, respectively for which $j<m_{i+1}$. In view of
item (2) in Corollary \ref{cor:S} these subsets form a partition of $\cS{n,k}{i,j}$. We will show that
\begin{equation}\label{eq:card_2}
  \#\cS{n,k}{i,j=m_{i+1}}=\#\cS{n-1,k}{i,j=m_{i+1}} \quad\mbox{ and }\quad
  \#\cS{n,k}{i,j<m_{i+1}}=\#\cS{n-1,k}{i,j<m_{i+1}}+\#\hcS{n-1,k}{i,j}\;,
\end{equation}%
which proves (\ref{eq:card_1}). First, let us consider the restriction of the injective map $\rho$ to
$\cS{n-1,k}{i,j=m_{i+1}}$. Its image is contained in $\cS{n,k}{i,j=m_{i+1}}$; in fact, its image consists of all of
$\cS{n,k}{i,j=m_{i+1}}$ since, as we pointed out in the proof of item (4), any element $\um'\in \cS{n,k}i$,
with $m'_{i+1}\leqs k$ belongs to the image of $\rho$. This proves the first equality in (\ref{eq:card_2}). For the
second equality, consider the following diagram:
$$
\xymatrix{
\cS{n-1,k}{i,j<m_{i+1}}\ar[rr]^-{\rho}\ar[d]_-{\tau_{n-1}}& &\cS{n,k}{i,j<m_{i+1}}\ar[d]^-{\tau_{n}}\\
\hcS{n-1,k}{i,j}\ar[rr]^-{\hat\rho_j}_-{\simeq}& &\hcS{n,k}{i,j}}
$$
The maps $\tau_n$ in it are defined by $\um\mapsto\hum$. Clearly this diagram is commutative. The lower line is a
bijection in view of (4) and the upper line is injective. We claim that for every element $\hum$ of
$\hcS{n-1,k}{i,j}$,
\begin{equation}\label{eq:diff}
  1+\#\tau_{n-1}^{-1}\set{\hum}=\#\tau_{n}^{-1}\set{\hat\rho_j(\hum)}\;.
\end{equation}%
Indeed, according to Corollary \ref{cor:S} (3) and (4), given $\hum\in\hcS{n-1,k}{i,j}$ the $\um$ such that
$\tau_{n-1}(\um)=\hum$ are precisely those for which $m_{i+1}$ satisfies the inequalities $m_{i+1}<m_1+n-1-k$ and
$m_{2i+1}<m_{i+1}+n-1-k$, so there are $2(n-1-k)+m_1-m_{2i+1}-1$ of them. Therefore
\begin{eqnarray*}
  &\#\tau^{-1}_{n-1}\set{\hum}=2(n-1-k)+m_1-m_{2i+1}-1=2(n-k)+m_1-m_{2i+1}-3\;,  &\\
  &\#\tau^{-1}_{n}\set{\hat\rho_j(\hum)}=2(n-k)+m_1-m'_{2i+1}-1=2(n-k)+m_1-m_{2i+1}-2\;.&
\end{eqnarray*}
This proves (\ref{eq:diff}). From it, the second equality in (\ref{eq:card_2}) is clear, because $\rho$ is
injective and the maps $\tau_n$ are surjections.

The proof of item (6) will be given at the end of the section.
\end{proof}
We now show our main claim, to wit that the matrix $\Lambda_{k,k}$ is of maximal rank. We do this by analyzing the
structure of this matrix. As before, $k$ and $n$ are fixed and $n>2k+1$. Consider the integers $\s_{i,j}^{(k)}$,
which we view as the entries of a matrix (of size $k\times k$). For fixed $i$, the right hand side of
(\ref{eq:rec_for_sigma}) is zero for $i-1$ consecutive \emph{integer} values of $j$, namely for
$\[\frac{k-i+3}2\]\leqs j\leqs \[\frac{k+i-1}2\]$ and it is positive for all other integer values of $j$, because
for such values of $j$, either all factors are negative or all factors are positive, and the number of factors is
$2i-2$, hence even.  This means that the integers $\s_{i,j}^{(k)}$ verify the following properties: for all
$i\in\{1,2,\ldots,k\}$ and $j\in\{1,2,\ldots,k-1\}$
\begin{equation*}
  \s^{(k)}_{i,j}\geqs \s^{(k)}_{i,j+1} \text{ with equality iff}
  \left[\frac{k-i+3}{2}\right]\leq j< \left[\frac{k+i+1}{2}\right].
\end{equation*}
The entries of $\cK^{(n,k)}$ enjoy the same property: in view of items (4) and (5) of Proposition
(\ref{prp:indep_tech}), $\cK^{(n,k)}_{i,j}-\cK^{(2k+1,k)}_{i,j}=(n-2k-1)\s^{(k)}_{i,j}$; also
$\cK^{(2k+1,k)}_{i,j}$ is independent of $j$ (see example \ref{example_bogo_case_S_i,j_same_cardinality}),
so that
\begin{equation*}
  \cK^{(n,k)}_{i,j}\geqs \cK^{(n,k)}_{i,j+1} \text{ with equality iff}
  \left[\frac{k-i+3}{2}\right]\leq j< \left[\frac{k+i+1}{2}\right].
\end{equation*}
In fact, according to item (2) of Proposition (\ref{prp:indep_tech}), the entries of $\Lambda=\Lambda_{k,k}$ also
enjoy the same property, since the cited item says that the lines of $\cK^{(n,k)}$ and $\Lambda$ are the same, up
to an additive constant. So the matrix $\Lambda$ has the following structure:
\begin{equation*}
  \Lambda_{i,j}\geqs \Lambda_{i,j+1} \text{ with equality iff}
  \left[\frac{k-i+3}{2}\right]\leq j< \left[\frac{k+i+1}{2}\right].
\end{equation*}
It follows that $\Lambda$ is of maximal rank. Indeed, the above property says that the entries of the $i$-th row
are decreasing in $j$ with exactly $i$ elements (in the middle) equal. Furthermore if the elements
$\Lambda_{i,j_1},\ldots,\Lambda_{i,j_i}$ of the $i$-th row are also equal, then the elements
$\Lambda_{i+1,j_1},\ldots,\Lambda_{i+1,j_i}$, of the $i+1$-th row are equal.  Subtracting a suitable multiple of
the last line (whose elements are all equal, and different from zero) from line $i$ we can make the $i$ equal
elements of line $i$ zero. Doing this for $i=1,\dots,k-1$ and rearranging the columns, we obtain a lower triangular
matrix with non-zero diagonal elements. This shows that $\Lambda$ is of maximal rank, and hence terminates ---
modulo the proof of item (6) in Proposition \ref{prp:indep_tech} --- the proof of Proposition
\ref{prp:independence}.

In order to prove item (6) of Proposition \ref{prp:indep_tech} we need two recurrence relations for
$\s^{(k)}_{i,j}$ which we prove in the next lemma.
\begin{lemma}\label{lem:recur_rel}
Let $k\geq i\geq 2$ and $k\geq j\geq 2$. Then
\begin{enumerate}
\item
$
\s^{(k)}_{i,1}=\s^{(k-1)}_{i-1,1}+2\s^{(k-2)}_{i-1,1}+\ldots+(k-i+1)\s^{(i-1)}_{i-1,1}
$
\item
$
\s^{(k)}_{i,j}=\s^{(k-1)}_{i,j-1}+\s^{(k-1)}_{i-1,j-1}+\s^{(k-2)}_{i-1,j-2}+\ldots+\s^{(k-j+1)}_{i-1,1}.
$
\end{enumerate}
\end{lemma}
\begin{proof}
Using item (1) of Proposition \ref{lma:S} we deduce that for any
$1\leqs m_1'< m_2'<\ldots<m_i'\leqs k$ such that $j\in\{m_1',\ldots,m_i'\}$ and any $n\geq 2k+1$ we have
$$
\#\{\hum\in\hcS{n,k}{i,j}:m_\ell=m_\ell'\text{ for }\ell=1,2,\ldots,i\}=(m_2'-m_1')(m_3'-m_2')\cdots(k+1-m_i').
$$
Therefore
\begin{equation}
\label{eqt:formula_for_the_sets_hcS}
\#\hcS{n,k}{i,j}=\sum_{\substack{1\leqs m_1<m_2<\ldots<m_i\leq k\\j\in\{m_1,m_2,\ldots,m_i\}}}
(m_2-m_1)(m_3-m_2)\cdots(k+1-m_i).
\end{equation}
From this formula we see that $\#\hcS{n,k}{i,j}$ is independent of $n$, which also follows from item (4) of
Proposition \ref{prp:indep_tech}. Therefore the choice $n=2k+1$ is reasonable.
Since for $\hum\in\hcS{2k+1,k}{i,1}$ we have that $m_1=1$, it follows
$$
\#\{\um\in\hcS{2k+1,k}{i,1}:m_i=\ell\}=(k+1-\ell)\s^{(\ell-1)}_{i-1,1}
$$
for all $i\leqs\ell\leqs k$.
Partitioning the set $\hcS{2k+1,k}{i,1}$ as
$$
\hcS{2k+1,k}{i,1}=\cup_{\ell=i}^k\{\um\in\hcS{2k+1,k}{i,1}:m_i=\ell\}
$$
we get the proof of item (1).

For the proof of item (2) first note that there is a correspondence
between the sets
$$
\{(m_1,m_2,\ldots,m_i):2\leqs m_1<m_2<\ldots<m_i\leq k\text{ and }j\in\{m_1,m_2,\ldots,m_i\}\}
$$
and
$$
\{(m_1,m_2,\ldots,m_i):1\leqs m_1<m_2<\ldots<m_i\leq k-1\text{ and }j-1\in\{m_1,m_2,\ldots,m_i\}\}.
$$
The correspondence is given by the function
$$
(m_1,m_2,\ldots,m_i)\mapsto(m_1-1,m_2-1,\ldots,m_i-1).
$$
Using the formula (\ref{eqt:formula_for_the_sets_hcS}) we get that
$\#\{\hum\in\hcS{2k+1,k}{i,j}:m_1\neq1\}=\s^{(k-1)}_{i,j-1}.$
Now we analyze the case $m_1=1$. For any $\ell\in\{2,3,\ldots,j\}$, formula (\ref{eqt:formula_for_the_sets_hcS})
gives
\begin{gather*}
\#\{\um\in\hcS{2k+1,k}{i,j}:m_1=1,m_2=\ell\}=\\
(\ell-1)\sum_{\substack{\ell<m_3<\ldots<m_i\leq k\\j\in\{\ell,m_3,\ldots,m_i\}}}
(m_3-\ell)(m_4-m_3)\cdots(k+1-m_i)
\end{gather*}
and therefore
\begin{gather}
\label{eqt:first_recur_formula_for_the_sets_hcS}
\begin{split}
\#\{\um\in\hcS{2k+1,k}{i,j}:m_1=1\}=&\\
\sum_{\ell=2}^j(\ell-1)\sum_{\substack{\ell<m_3<\ldots<m_i\leq k\\j\in\{\ell,m_3,\ldots,m_i\}}}
(m_3-\ell)(m_4-m_3)\cdots&(k+1-m_i).
\end{split}
\end{gather}
This is because $m_2$ can only take the values $2,3,\ldots,j$.
In a similar manner as in the case $m_1\neq1$, for $2\leqs j'\leqs j$
we have
\begin{gather*}
\sum_{\ell=j'}^j\sum_{\substack{\ell<m_3<\ldots<m_i\leq k\\j\in\{\ell,m_3,\ldots,m_i\}}}
(m_3-\ell)(m_4-m_3)\cdots(k+1-m_i)=\\
\sum_{\ell=1}^{j-j'+1}\sum_{\substack{\ell<m_3<\ldots<m_i\leq k-j'+1\\j-j'+1\in\{\ell,m_3,\ldots,m_i\}}}
(m_3-\ell)(m_4-m_3)\cdots(k-j'+2-m_i)
\end{gather*}
which is exactly $\s^{(k-j'+1)}_{i-1,j-j'+1}$. This is because the first entry of any vector in
$\hcS{2k-2j'+3,k-j'+1}{i-1,j-j'+1}$ can only take the values $m_1=1,2,\ldots,j-j'+1$.
Combining with formula (\ref{eqt:first_recur_formula_for_the_sets_hcS}) and the case
$m_1\neq 1$, we get item (2).

\end{proof}
Before giving the general proof of item (6) of Proposition \ref{prp:indep_tech}.
we will prove it for the special cases $i=1$ and $i=k$. We do this in the following example.
\begin{example}
\label{example_casei=1_and_i=k}
For the case $i=1$ we easily get (see for example (\ref{eqt:formula_for_the_sets_hcS}) or Proposition \ref{lma:S})
$$
\s^{(k)}_{1,j}=k-j+1, \; j=1,2,\ldots,k.
$$
Therefore $\s^{(k)}_{1,j}-\s^{(k)}_{1,j+1}=1$.
For the case $i=k$, we already pointed out that the set $\hcS{2k+1,k}{k}$ has exactly one element,
namely $\hcS{2k+1,k}{k}=\{(1,2,\ldots,k,k+2,k+3,\ldots,2k+1)\}$.
It follows that $\s^{(k)}_{k,j}=1, \; j=1,2,\ldots,k$,
and the sequence $\s^{(k)}_{k,j}$ is constant.
\end{example}
\begin{proof}[Proof of item (6) of Proposition \ref{prp:indep_tech}]
First we prove (using induction on $k$) that
\begin{equation}
\s^{(k)}_{i,1}=\binom{k+i-1}{2i-1}, \text{ for all } 1\leq i\leq k.
\end{equation}
For $i=1$ this formula says that $\s^{(k)}_{1,1}=k$, which is included in the Example \ref{example_casei=1_and_i=k}.
Assuming that $\s^{(k')}_{i,1}=\binom{k'+i-1}{2i-1}$ for all $1\leqs k'<k$ and all $i\leq k'$,
then using the recurrence relation of item (1) of Lemma \ref{lem:recur_rel} we get
$$
\s^{(k)}_{i,1}=\binom{k+i-3}{2i-3}+2\binom{k+i-4}{2i-3}+\dots+(k-i+1)\binom{2i-3}{2i-3}=\binom{k+i-1}{2i-1}.
$$

For our proof we will also use induction.  For $k=1$ and $k=2$ the proof is in the Example
\ref{example_casei=1_and_i=k}.  We suppose $k>2$ and we consider the case $2j-k+i-2\geq 0$ (the case $2j-k+i-2<0$
being the same).  In this case we will show that $\s^{(k)}_{i,j}-\s^{(k)}_{i,j+1}=\binom{2j-k+i-2}{2i-2}$ (in the
case $2j-k+i-2<0$ we have to show that $\s^{(k)}_{i,j}-\s^{(k)}_{i,j+1}=\binom{-2j+k+i-1}{2i-2}$).  Assuming the
truth of this formula for $k'<k$, then using the recurrence relation of item (2) of Lemma \ref{lem:recur_rel} we
get
\begin{gather*}
\s^{(k)}_{i,j}-\s^{(k)}_{i,j+1}=\s^{(k-1)}_{i,j-1}-\s^{(k-1)}_{i,j}+
\sum_{\ell=1}^{j-1}(\s^{(k-j+\ell)}_{i-1,\ell}-\s^{(k-j+\ell)}_{i-1,\ell+1})-\s^{(k-j)}_{i-1,1}=\\
\binom{2j-k+i-3}{2i-2}+\sum_{\ell=1}^{j-1}\binom{\ell-k+j+i-3}{2i-4}-\binom{k-j+i-2}{2i-3}=\\
\binom{2j-k+i-3}{2i-2}+\binom{2j-k+i-3}{2i-3}-\binom{-k+j+i-2}{2i-3}-\binom{k-j+i-2}{2i-3}=\\
\binom{2j-k+i-2}{2i-2}
\end{gather*}
which follows from the formulas
$$
\binom{2j-k+i-3}{2i-2}+\binom{2j-k+i-3}{2i-3}=\binom{2j-k+i-2}{2i-2}
$$
and
$$
\binom{-k+j+i-2}{2i-3}=\binom{k-j+i-2}{2i-3}.
$$
\end{proof}

\section{Non-commutative and Liouville integrability}\label{sec:integ}
In this section, we use the results of the previous section to prove our main result, Theorem \ref{thm:main}, which
states that the Lotka-Volterra systems $\LV(n,k)$ (with $n>2k+1$) are Liouville integrable as well as
non-commu\-tative integrable (of rank $k+1$). First, let us recall the following definition.
\begin{defn}\label{def:non-com}
Let $(M,\Pi)$ be a Poisson manifold of dimension $n$. Let $\F=(f_1,\dots,f_s)$ be an $s$-tuple of functions on $M$,
where $2s\geqs n$ and set $r:=n-s$. Suppose the following:
\begin{enumerate}
  \item[(1)] The functions $f_1,\dots,f_r$ are in involution with the functions $f_1,\dots,f_s$:
  $$ \{f_i,f_j\}=0,\qquad (1\leqs i\leqs r \hbox{ and } 1\leqs j \leqs s)\;;$$
  \item[(2)] For $m$ in a dense open subset of $M$:
  $$ \diff_m f_1\wedge \dots \wedge \diff_m f_s\ne 0 \quad \hbox{and} \quad \X_{f_1}|_m\wedge \dots \wedge
    \X_{f_r}|_m\ne 0\;.  $$
\end{enumerate}
Then the triplet $(M,\Pi,\F)$ is called a \emph{non-commutative integrable system} of \emph{rank} $r$.
\end{defn}
The classical case of a \emph{Liouville integrable system} corresponds to the particular case where $r$ is half the
(maximal) rank of $\Pi$; this implies that \emph{all} the functions $f_1,\dots,f_s$ are pairwise in involution. The
case of a superintegrable system corresponds to $r=1$; in the latter case, the Poisson structure does not play any
role.

We first consider the non-commutative integrability (of rank $k+1$) of $\LV(n,k)$ (with $n>2k+1$). The $n-k-1$
first integrals which we consider are $H=K_0^{(n,k)},K_1^{(n,k)},\dots,K_k^{(n,k)}$ (see Subsection
\ref{subsec:polynomial_integrals}) and $H^{(n,k)}_1,H^{(n,k)}_2,\dots,H^{(n,k)}_{n-2k-2}$ (see Subsection
\ref{subsec:rational_integrals}).  We know already from the previous sections that all functions are first
integrals of $\LV(n,k)$. Notice that when $n$ is odd, $H^{(n,k)}_1$ is just the Casimir function $C$ (see
(\ref{eq:Casimir})). It was shown by Itoh (see \cite{itoh2}) that the functions $K_i$ are in involution, hence the
functions $K_i^{(n,k)}$ are also in involution, being restrictions to a Poisson submanifold. We show in the
following proposition that the functions $K_i^{(n,k)}$ are in involution with the functions $H^{(n,k)}_\ell$.

\begin{prop}
  For $\ell=1,\dots,n-2k-2$ and for $i=1,\dots,k$ the functions $K_i^{(n,k)}$ and $H^{(n,k)}_\ell$ are in
  involution.
\end{prop}
\begin{proof}
We give the proof for $n$ odd and we write $\PB$ for $\PB_k^{(n)}$. By using the involution $\psi$, if necessary,
we may assume that $1\leqs \ell\leqs \frac{n-1}2+k$. Suppose that $X$ is a polynomial in $x_1,\dots,x_n$ of the
form $X=(L+L')Y$, where $L$ and $L'$ are linear, $Y$ and $L'$ are independent of the variables
$x_{k+1},\dots,x_{n-k}$ and $L$ is the sum of these variables. We will show that $\pb{H^{(n,k)}_\ell,X}=0$; since
by items (3) and (4) of Corollary \ref{cor:S}, $K_i^{(n,k)}$ is a (finite) sum of terms of this form, it follows
that $\pb{H^{(n,k)}_\ell,K_i^{(n,k)}}=0$, which was to be shown. First, since every variable which appears in $L'$
or $Y$ also appears in each term of $H_\ell^{(n,k)}$ (in fact it appears in $\hat H_\ell^{(n,k)}$), we know by item
(1) in Lemma \ref{lma:xs} that $\pb{H^{(n,k)}_\ell,L'Y}=0$. It remains to be shown that $\pb{H^{(n,k)}_\ell,LY}=0$,
where we recall that $L=x_{k+1}+\dots+x_{n-k}$. Using again item (1) of Lemma \ref{lma:xs} in the two first
equalities that follow, and item (2) of the same lemma in the fourth equality, we get
\begin{eqnarray*}
  \pb{H_\ell^{(n,k)},Y\sum_{s=k+1}^{n-k}x_s}&=&Y\pb{H_\ell^{(n,k)},\sum_{s=k+1}^{n-k}x_s}
      =Y\pb{H_\ell^{(n,k)},\sum_{s=k+1}^{n+2\ell-1}x_s}\\
      &=&Y\sum_{j,s=k+1}^{n+2\ell-1}\pb{x_j\hat H_\ell^{(n,k)},x_s}
      =Y\hat H_\ell^{(n,k)}\sum_{j,s=k+1}^{n+2\ell-1}\pb{x_j,x_s}\;.
\end{eqnarray*}
The latter sum is zero because of skew-symmetry of the Poisson bracket.
\end{proof}
To finish the proof of non-commutative integrability, it remains to be shown that the Hamiltonian vector fields,
associated to the $k+1$ first integrals $K_0^{(n,k)},\dots,K_k^{(n,k)}$ are independent on an open dense subset on
$\bbR^n$. When $n$ is even, the Poisson structure is symplectic, and so this follows from the functional
independence of $K_0^{(n,k)},\dots,K_k^{(n,k)}$. When $n$ is odd, the Poisson structure is of rank $n-1$ and a
Casimir is given by the rational function $H^{(n,k)}_1$, and so the functional independence of
$H^{(n,k)}_1,K_0^{(n,k)},K_1^{(n,k)}\dots,K_k^{(n,k)}$ leads to the same conclusion.

Let us now consider Liouville integrability. We know from \cite{KKQTV} that the functions $F_1,\dots,F_{r-1}$ are
in involution, with $r:=\left[\frac{n+1}2\right]-k$. According to Proposition \ref{prop:poisson_map}, $\phi_k$ is a
Poisson map, and so the pullbacks $H^{(n,k)}_1,\dots,H^{(n,k)}_{r-1}$ are also pairwise in involution. The upshot
is that the $\left[\frac{n+1}2\right]$ independent functions $H^{(n,k)}_1,\dots,H^{(n,k)}_{r-1},$
$K_0^{(n,k)},\dots,K_k^{(n,k)}$ are in involution. Since $\pi_k^{(n)}$ is of rank $n$ when $n$ is even and of rank
$n-1$ when $n$ is odd, this proves Liouville integrability. Notice that, rather than using the functions
$F_1,\dots,F_{r-1}$, one can also use the functions $G_1,\dots,G_{r-1}$, because they are also in involution (since
$\psi$ is an anti-Poisson map).

This finishes the proof of Theorem \ref{thm:main}.

\bibliographystyle{abbrv}


\end{document}